\documentclass[letterpaper, 10 pt, conference]{ieeeconf}  

\IEEEoverridecommandlockouts                              

\title{\LARGE \bf
Differentially Private Algorithms for Synthetic Power System Datasets
}

\author{Vladimir Dvorkin and Audun Botterud
\thanks{This work is supported by the Marie Sklodowska-Curie Actions and Iberdrola Group, Grant Agreement \textnumero101034297 – project Learning ORDER.}
\thanks{Vladimir Dvorkin and Audun Botterud are with the Laboratory for Information \& Decision Systems, Massachusetts Institute of Technology (MIT),
Cambridge, MA 02139, USA. Vladimir Dvorkin is also with the MIT Energy Initiative. {\tt\footnotesize \{dvorkin,audunb\}@mit.edu}
}%
}

\usepackage{graphicx}
\usepackage{amsmath}
\usepackage[bb=boondox]{mathalfa}
\usepackage{amssymb}
\usepackage{url}

\usepackage{algorithm2e}
\usepackage{xcolor}

\newcommand{\xvbox}[2]{\makebox[#1][l]{#2}}

\SetAlFnt{\small}

\SetCommentSty{xCommentSty}

\SetNlSty{mynlfont}{}{} 

\LinesNumbered

\SetSideCommentRight

\DontPrintSemicolon

\RestyleAlgo{algoruled}

\let\oldnl\nl
\newcommand{\nonl}{\renewcommand{\nl}{\let\nl\oldnl}}
\usepackage{stmaryrd}

\usepackage{amsfonts}
\usepackage{enumerate}
\usepackage[utf8]{inputenc}
\usepackage[top=1 in,bottom=1in, left=0.5 in, right=0.5 in]{geometry}
\usepackage{euscript}
\usepackage{empheq}
\usepackage{mathtools}
\usepackage{nicefrac}
\usepackage{xfrac}
\usepackage{multirow}
\usepackage{booktabs}

\newtheorem{assumption}{Assumption}
\newtheorem{definition}{Definition}
\newtheorem{theorem}{Theorem}
\newtheorem{remark}{Remark}
\newtheorem{lemma}{Lemma}

\usepackage{tikz}
\usepackage{pgfplots}
\usetikzlibrary{backgrounds}
\definecolor{maincolor}{HTML}{032F99}
\newcommand{\boundellipse}[3]
{(#1) ellipse [x radius=#2,y radius=#3]
}
\usetikzlibrary{intersections}

\newcommand{\minimize}[1]{\underset{{#1}}{\text{minimize}}}
\newcommand{\maximize}[1]{\underset{{#1}}{\text{maximize}}}
\newcommand{\argmin}[1]{\text{arg}\underset{#1}{\text{min}}}

\newcommand{\st}{\text{subject to}}

\newcommand\norm[1]{\left\lVert#1\right\rVert}

\usepackage{filecontents,pgfplots}
\pgfplotsset{compat=1.10}
\usepgfplotslibrary{fillbetween}

\usepackage{xcolor}
\definecolor{five}{rgb}{0.39,0.58,0.69} 
\definecolor{fifteen}{rgb}{0.0,0.35,0.58} 
\definecolor{thirty}{rgb}{0.01,0.22,0.42}

\usepackage[font=footnotesize,skip=2.5pt]{caption}

\begin{document}

\maketitle
\thispagestyle{empty}
\pagestyle{empty}

\begin{abstract}
While power systems research relies on the availability of real-world network datasets, data owners (e.g., system operators) are hesitant to share data due to security and privacy risks. To control these risks, we develop privacy-preserving algorithms for the synthetic generation of optimization and machine learning datasets. Taking a real-world dataset as input, the algorithms output its noisy, synthetic version, which preserves the accuracy of the real data on a specific downstream model or even a large population of those. We control the privacy loss using Laplace and Exponential mechanisms of differential privacy and preserve data accuracy using a post-processing convex optimization. We apply the algorithms to generate synthetic network parameters and wind power data.
\end{abstract}

\begingroup
\allowdisplaybreaks

\setlength{\abovedisplayskip}{3.5pt}
\setlength{\belowdisplayskip}{3.5pt}
\setlength{\abovedisplayshortskip}{3.5pt}
\setlength{\belowdisplayshortskip}{3.5pt}

\setlength{\belowcaptionskip}{-10pt}

\section{Introduction}

Power system datasets are instrumental for enhancing solutions to many problems, including optimal power flow (OPF) and wind power forecasting. Releasing real data, however, is challenging due to security and privacy concerns. For example, detailed network datasets inform false data injection attacks on SCADA systems \cite{jin2018power}, and strategic market players may leverage bidding records to maximize profits at the expense of deteriorating social welfare \cite{chen2019learning}. These concerns motivate producing synthetic datasets -- a sanitized version of private datasets that approximately preserve accuracy of data for power system problems. 

Differential privacy (DP) is an algorithmic notion of privacy preservation that enables quantifiable trade-offs between data privacy and accuracy \cite{dwork2014algorithmic}. It has found applications in the context of privacy-preserving OPF computations,  e.g., in distributed control algorithms \cite{dvorkin2020differentially} and in centralized solvers for distribution and high-voltage grids \cite{dvorkin2020adifferentially,zhou2019differential}. It has also been applied to enhance data privacy in machine learning problems in power systems \cite{dvorkin2022privacy}. Models in \cite{dvorkin2020differentially}--\cite{dvorkin2022privacy}, however, only control data leakages in computational results and do not provide synthetic data per se. 

Producing synthetic datasets in a DP way is achieved by corrupting data with privacy-preserving noise \cite{vietri2020new,hardt2012simple}. However, applications of the standard noise-additive DP mechanisms in power systems, such as the Laplace mechanism, may no longer admit a meaningful result. Indeed, adding noise to data may fundamentally alter important statistics and trends in machine learning datasets \cite{wang2022privacy}, e.g., monotonic dependency of wind power generation on wind speed. In the OPF context, the authors in \cite{fioretto2019differential} and \cite{mak2019privacy} showed that the Laplacian perturbation of network parameters almost surely violates feasibility on a broad range of power system benchmarks. As a remedy, they proposed an optimization-based post-processing which restores the accuracy of synthetic OPF datasets without altering the privacy guarantee. The proposed restoration, however, renders the synthetic dataset feasible only for a particular OPF model. Repeated applications of the Laplace mechanism to restore accuracy on many OPF models (e.g., for different instances of variable renewable production) may not be possible, as noise must be scaled respecting the number of repetitions, as per composition of DP \cite{dwork2014algorithmic}. 

In this letter, we introduce two private synthetic dataset generation algorithms for power systems, which ensure the accuracy of synthetic datasets for downstream models. The algorithms enjoy a combination of known DP mechanisms and convex (and mixed-integer) optimization of synthetic data. Specifically, we develop:
\begin{enumerate}
    \item Wind power obfuscation (WPO) algorithm which privately releases historical wind power measurements, while guaranteeing DP of the real data and ensuring accuracy in terms of the outcomes of a downstream regression analysis. 
    \item Transmission capacity obfuscation (TCO) algorithm, which releases synthetic line parameters, while ensuring that they remain feasible and cost-consistent with respect to real data on a population of OPF models. Here, we use both Laplace and Exponential mechanisms of DP to substantially reduce the noise compared to using the Laplace mechanism alone. 
\end{enumerate}

In the next section, we review the basic DP results. In Sections \ref{sec:wpo} and \ref{sec:TCO} we present the two algorithms and their theoretical properties. Section \ref{sec:experiments} provides numerical experiments, and Section \ref{sec:conclusions} concludes. Proofs are relegated to the Appendix. 




\section{Preliminaries on Differential Privacy}\label{sec:preliminaries}
This section reviews basic DP results serving as building blocks for our privacy-preserving dataset generation algorithms.  

Consider a vector $y\in\mathcal{Y}\subseteq\mathbb{R}^{n}$ collecting $n$ private records from a dataset universe $\mathcal{Y}$, and consider a query $Q:\mathcal{Y}\mapsto\mathcal{R}$ as a mapping from universe $\mathcal{Y}$ to some range $\mathcal{R}$. Queries of interest include simple numerical queries, i.e., identity query $Q(y)=y$, and optimization and ML queries, such as OPF or regression models. The goal is to make {\it adjacent} vectors $y,y'\in\mathcal{Y}$ of  private records, statistically indistinguishable in query answers. 
\begin{definition}[Adjacency \cite{chatzikokolakis2013broadening}]\label{def:adjacency} Vectors $y,y'\in\mathcal{Y}$ are said to be $\alpha-$adjacent, denoted as $y\sim_{\alpha}y'$, if $\exists i\in1,\dots,n$, s.t. $y_{j}=y_{j}^{\prime},\forall j\in\{1,\dots,n\}\backslash i$, and $|y_{i}-y'_{i}|\leqslant\alpha$ for $\alpha>0$. That is, the adjacent datasets are different only in one item by at most $\alpha$. 
\end{definition}

A statistical similarity of query answers is captured by the notion of differential privacy, attained through randomization. 

\begin{definition}[$\varepsilon-$differential privacy \cite{dwork2014algorithmic}] A random query $\tilde{Q}:\mathcal{Y}\mapsto\mathcal{R}$ is $\varepsilon-$differentially private if, for any output $r\subseteq\mathcal{R}$ and any $\alpha-$adjacent vectors $y,y'\in\mathcal{Y}$, the following ratio holds
\begin{align}
    \frac{\text{Pr}[\tilde{Q}(y\textcolor{white}{'})=r]}{\text{Pr}[\tilde{Q}(y')=r]}\leqslant\text{exp}(\varepsilon).
\end{align}
where probability is with respect to the randomness of $\tilde{Q}$.
\end{definition}

Privacy parameter $\varepsilon>0$ is termed {\it privacy loss}: with smaller $\varepsilon$ we achieve stronger privacy protection. Indeed, for small $\varepsilon$ we have $\text{exp}(\varepsilon)\approx1+\varepsilon$, thereby making any two adjacent datasets $y$ and $y'$ statistically similar in the answer of the randomized query. 
\begin{theorem}[Sequential composition \cite{dwork2014algorithmic}]\label{th:SC} A series $\tilde{Q}_{1}(y),\dots,$ $\tilde{Q}_{k}(y)$ of $\varepsilon_{i}-$DP queries on dataset $y$ satisfies $\sum_{i=1}^{k}\varepsilon_{i}-$DP.
\end{theorem}

\begin{theorem}[Post-processing immunity \cite{dwork2014algorithmic}]\label{th:PP} If query $\tilde{Q}$ satisfies $\varepsilon$-DP, then $g\circ\tilde{Q}(y)$, where $g$ is an arbitrary, data-independent post-processing of the query answer, also satisfies $\varepsilon$-DP. 
\end{theorem}

The first results bounds the privacy loss over multiple queries, and the second result states that any data-independent transformation of a DP query answer preserves the privacy guarantee.  

A numerical query is made DP by adding random noise to its output. The noise magnitude depends on the worst-case sensitivity $\delta_{Q}$ of query $Q$ to adjacent datasets, defined as 
\begin{align*}
    \delta_{Q} = \text{max}_{y\sim_{\alpha}y'}\norm{Q(y)-Q(y')}_{1}.
\end{align*}
Let $\text{Lap}(\lambda)^{k}$ denote a sample from the $k-$dimensional Laplace distribution with zero mean and scale parameter $\lambda$. DP of a numerical query is then achieved with the following result.  
\begin{theorem}[Laplace mechanism \cite{dwork2006calibrating}]\label{th:LM} Let $Q$ be a query that maps datasets to $\mathbb{R}^{k}$. Then, the Laplace mechanism which outputs $Q(y)+\text{Lap}(\delta_{Q}/\varepsilon)^{k}$ achieves $\varepsilon-$DP.
\end{theorem}

We also like to limit privacy losses when answering non-numerical queries. For example, given a population $\mathcal{Q}$ of queries, we would like to answer the question: {\it which query $Q\in\mathcal{Q}$ gives the maximum value on a private dataset $y$?} The following Exponential mechanism answers this question privately.

\begin{theorem}[Exponential mechanism \cite{mcsherry2007mechanism}]\label{th:EM} Let $\mathcal{Q}$ be a query population, and let $u:\mathcal{Y}\times\mathcal{Q}\mapsto\mathbb{R}$ be the score function with sensitivity $\delta_{u}$. Then, the Exponential mechanism which outputs query $Q\in\mathcal{Q}$ proportionally to $\text{exp}\left(\frac{\varepsilon u(y,Q)}{2\delta_{u}}\right)$ attains $\varepsilon-$DP.
\end{theorem}

For discrete populations of queries, i.e.,  $\mathcal{Q} = Q_{1},\dots,Q_{m}$, we can adopt the report-noisy-max algorithm \cite[\S3.3]{dwork2014algorithmic} -- an efficient implementation of the exponential mechanism for finite $\mathcal{Q}$. 

Next, we leverage these results to design DP algorithms for synthetic dataset generation as applicable to power systems.

\section{Privacy-Preserving Wind Power Dataset Release}\label{sec:wpo}

Consider the problem of a wind turbine operator (data owner) who wants to release synthetic wind power records in a differentially private way. The real dataset $\mathcal{D}=\{(x_{1},y_{1}),\dots,(x_{m},y_{m})\}$ consists of $m$ records, where each record $i$ includes some public weather data $x_{i}\in\mathbb{R}^{n}$ and a private power measurement $y_{i}\in\mathbb{R}$ subject to obfuscation. The release of the synthetic dataset takes the form $\tilde{\mathcal{D}}=\{(x_{1},\tilde{y}_{1}),\dots,(x_{m},\tilde{y}_{m})\}$, where $\tilde{y}_{i}$ is a synthetic measurement. To provide formal privacy guarantees in this release, we could perturb each real record $y_{i}$ with the Laplace mechanism of Theorem \ref{th:LM}. However, the application of the Laplace mechanism alone is ignorant of the accuracy of the resulting dataset in the downstream analysis, and such a release may not be useful. We discuss the dataset accuracy in terms of the outcomes of a regression (downstream) problem
\begin{align}
    \minimize{\beta}\quad&\norm{X\beta-y} + \lambda\norm{\beta} ,\label{base_reg}
\end{align}
which minimizes the loss function by optimally choosing regression weights $\beta\in\mathbb{R}^{p}$, given some small regularization parameter $\lambda$ to prevent overfitting. Here, matrix $X^{m\times p}$ collects weather features; we do not require $p=n$, as model \eqref{base_reg} may not include all meteorological data from $\mathcal{D}$ and may also enjoy certain feature transformations (e.g., squared wind speeds). The goal is thus to release a synthetic dataset $\tilde{\mathcal{D}}$ whose regression loss and weights are consistent with those on the real dataset. On a particular vector of measurements $\overline{y}$, we denote the regression loss and weights by $\ell(\overline{y})$ and $\beta(\overline{y})$, respectively. To estimate them on the real dataset privately, we need to bound their sensitivities to adjacent datasets.

\begin{lemma}[Regression sensitivity bounds]\label{lm:reg_sens} For any two $\alpha$-adjacent vectors of wind power measurements $y,y'\in\mathbb{R}^{m}$, the worst-case sensitivity of regression weights is bounded as 
\begin{align*}
    \delta_{\beta}=\text{max}_{y\sim_{\alpha}y'}\norm{\beta(y)-\beta(y')}_{1}\leqslant\norm{(X^{\top}X + \lambda I)^{-1}X^{\top}}_{1}\alpha,
\end{align*}
and the worst-case sensitivity of the regression loss $$\delta_{\ell}=\text{max}_{y\sim_{\alpha}y'}\norm{\ell(y) - \ell(y')}_{1}$$ is bounded by the solution of the following problem:
\begin{align*}
    \delta_{\ell}\leqslant\maximize{i=1,\dots,m}\;\norm{(X(X^{\top}X + \lambda I)^{-1}X^{\top} - I)(e_{i}\circ\alpha)}.
\end{align*}
\end{lemma}
\begin{proof}See Appendix \ref{app:lemma1}.
\end{proof}

Importantly, the two bounds only depend on public information, i.e., weather features, regularization and adjacency parameters, and completely independent from private measurements $y$.

\subsection{Differentially Private WPO Algorithm}

\begin{algorithm}[t]
\setlength{\abovedisplayskip}{0pt}
\setlength{\belowdisplayskip}{0pt}
\setlength{\abovedisplayshortskip}{0pt}
\setlength{\belowdisplayshortskip}{0pt}
\SetKwInOut{Input}{Input}
\SetKwInOut{Output}{Output}
\caption{Differentially private WPO}\label{alg:WPO}
\Input{$\text{WP records}\;\mathcal{D}=\{(x_{1},y_{1}),\dots,(x_{m},y_{m})\}$; \\ DP param. $\varepsilon_{1},\varepsilon_{2},\alpha$; regularization param. $\gamma_{\beta},\gamma_{y}$}
\BlankLine 
\Output{%
    \xvbox{4mm}{Synthetic WP records $\tilde{\mathcal{D}}=\{(x_{1},\tilde{y}_{1}),\dots,(x_{m},\tilde{y}_{m})\}$}
   }
\BlankLine 
Initialize synthetic measurements $\tilde{y}^{0} = y + \text{Lap}\left(\alpha/\varepsilon_{1}\right)^{m}$
\BlankLine 
Laplace mechanism to privately compute regression results 
\begin{align*}
    &\overline{\ell}=\ell(y) + \text{Lap}\left(\delta_{\ell}/\varepsilon_{2}\right)\quad\overline{\beta}=\beta(y) + \text{Lap}\left(\delta_{\beta}/\varepsilon_{2}\right)
\end{align*}
\BlankLine 
Post-processing optimization of $\tilde{y}^{0}$:
\begin{subequations}\label{PP_reg}
\begin{align}
    \tilde{y}\in\argmin{\tilde{y}}\quad& \norm{\overline{\ell} - \ell}+\gamma_{\beta}\norm{\overline{\beta} - \beta} + \gamma_{y}\norm{\tilde{y}^{0} - \tilde{y}}\label{PP_reg:obj}\\
    \st\quad
    &\mathbb{0} \leqslant \tilde{y} \leqslant \mathbb{1},\label{PP_reg:con_1}\\
    &\beta\in\argmin{\beta}\;\norm{X\beta-\tilde{y}} + \lambda\norm{\beta}\label{PP_reg:con_2}
\end{align}
\end{subequations}
\textbf{return:} synthetic wind power measurements $\tilde{y}$
\end{algorithm}
\setlength{\textfloatsep}{5pt}

We now introduce the privacy-preserving wind power obfuscation (WPO) Algorithm \ref{alg:WPO}. The algorithm takes the real dataset, privacy and regularization parameters as inputs, and produces a consistent synthetic dataset of wind power records. It relies on Lemma \ref{lm:reg_sens} to privately reveal regression results on a real dataset, and then leverages them to restore the consistency of the synthetic dataset using a post-processing optimization. Specifically, at Step 1, the algorithm initializes the synthetic datasets using the Laplace mechanism. Then, at Step 2, it computes the approximate regression loss and weights using the Laplace mechanism twice. At the last Step 3, the synthetic dataset undergoes optimization-based post-processing to ensure that the regression results on the synthetic dataset are consistent with those on the real data. 

The post-processing is based on the hierarchical optimization \eqref{PP_reg}, where the upper-level problem \eqref{PP_reg:obj}--\eqref{PP_reg:con_1} optimizes the synthetic dataset $\tilde{y}$ in response to the outcomes of the embedded lower-level regression problem \eqref{PP_reg:con_2}. In the upper-level objective, the first term improves the consistency in terms of regression loss, while the second and third terms are used for regularizing the synthetic dataset. Indeed, the losses $l$ and $\overline{l}$ can be matched with infinitely many assignments of $\beta$ and $\tilde{y}$. Thus, by setting a small parameter $\gamma_{\beta}>0$, the matching is achieved with a close approximation of the regression weights on the real data. Similarly, by setting a small parameter $\gamma_{y}>0$, we regularize the new data points according to the perturbation of real data points at Step 1. Finally, the upper-level constraint \eqref{PP_reg:con_1} guarantees that the synthetic dataset respects the nominal power limits.

While the hierarchical optimization \eqref{PP_reg} is originally intractable, we arrive at its tractable  convex reformulation by substituting the lower-level problem \eqref{PP_reg:con_2} with the following constraints:
\begin{subequations}\label{LL_ref}
\begin{align}
    &\beta = (X^{\top}X + \lambda I)^{-1}X^{\top}\tilde{y},\label{LL_ref_1}\\
    &\norm{X\beta - \tilde{y}}\leqslant\ell,\label{LL_ref_2}
\end{align}
\end{subequations}
where the linear constraint \eqref{LL_ref_1} is the closed-form solution to regression weights on vector $\tilde{y}$, and the conic constraint \eqref{LL_ref_2} is used to compute the loss on the same vector and weights. 

We now state the $\varepsilon-$DP guarantee of this algorithm.

\begin{theorem}[DP of the WPO Algorithm]\label{th:DP_WPO} Setting $\varepsilon_{1}=\varepsilon/2$ and $\varepsilon_{2}=\varepsilon/4$ renders Algorithm \ref{alg:WPO} $\varepsilon-$DP for $\alpha-$adjacent wind power datasets.  
\end{theorem}
\begin{proof}See Appendix \ref{app:th_WPO}.
\end{proof}

\section{Privacy-Preserving DC-OPF Dataset Release}\label{sec:TCO}

We now consider a problem of a power system operator who wants to release a synthetic network dataset in a differentially private way. The goal is to guarantee not only privacy but also accuracy with respect to possible downstream computations on the synthetic dataset. We consider the DC-OPF problem as the main computational task. We also specifically focus on the release of transmission capacity data, though other network parameters (loads, generation cost, etc.) can be released similarly.

The OPF problem models operations in a power network with $n$ buses and $e$ transmission lines. The goal is to compute the least-cost generation dispatch $p\in\mathbb{R}^{n}$ while satisfying electric loads $d\in\mathbb{R}_{+}^{n}$. Generators produce at linear costs $c\in\mathbb{R}_{+}^{n}$ within the minimum and maximum technical limits, encoded in set $\mathcal{P}=\{p\;|\;\underline{p}\leqslant p\leqslant \overline{p}\}.$ The DC power flows are modeled using the power transfer distribution matrix $F\in\mathbb{R}^{e\times n}$, and resulting power flows $\varphi=F(p-d)\in\mathbb{R}^{e}$ are limited by line capacities $\overline{f}\in\mathbb{R}_{+}^{e}$. 

Suppose that there is a set $1,\dots,m$ of OPF models, where each model $i$ includes a specific cost vector $c_i$, generation limits in set $\mathcal{P}_{i}$, and electric loads $d_{i}$. The transmission data, i.e., topology encoded in $F$ and capacity $\overline{f}$, remain the same. Each OPF model $i$ is then described by a tuple $\langle c_{i},d_{i},\mathcal{P}_{i},F,\overline{f}\rangle.$ Given the real OPF dataset $\langle c_{i},d_{i},\mathcal{P}_{i},F,\overline{f}\rangle_{i=1}^{m}$, the goal is to produce its synthetic version $\langle c_{i},d_{i},\mathcal{P}_{i},F,\overline{\varphi}\rangle_{i=1}^{m}$ with an obfuscated transmission capacity  vector $\overline{\varphi}$, which permits feasible and cost-consistent -- with respect to real data -- OPF outcomes across $m$ models.

Towards the goal, we formulate a DC-OPF problem parameterized by the synthetic transmission capacity $\overline{\varphi}$: 
\begin{subequations}\label{OPF_base}
\begin{align}
    \mathcal{C}_{i}(\overline{\varphi})=\minimize{p\in\mathcal{P}_{i}}\quad& c_{i}^{\top}p\label{OPF_base:obj}\\
    \st\quad
    &\mathbb{1}^{\top}(p-d_{i})=0,\label{OPF_base:bal}\\
    &\norm{F(p-d_{i})}_{1}\leqslant \overline{\varphi}\label{OPF_base:flo},
\end{align}
\end{subequations}
where the objective function \eqref{OPF_base:obj} minimizes OPF costs, denoted by $\mathcal{C}_{i}(\overline{\varphi})$, subject to power balance \eqref{OPF_base:bal}, flow and generation limits in \eqref{OPF_base:flo} and $\mathcal{P}_{i}$, respectively; all specific to a particular model $i$. We make two assumptions on problem \eqref{OPF_base}. 
\begin{assumption}[Feasibility]\label{ass:feasibility}
$\mathcal{C}_{i}(\overline{f})$ exists for all $i=1,\dots,m$.
\end{assumption}
\begin{assumption}[Sensitivity]\label{ass:sensitivity} Let $\overline{\varphi}_{1}\sim_{\alpha}\overline{\varphi}_{2}$ be two $\alpha-$adjacent vectors of transmission capacities. Then, $\norm{\mathcal{C}_{i}(\overline{\varphi}_{1}) - \mathcal{C}_{i}(\overline{\varphi}_{2})}_{1}\leqslant\overline{c}\alpha,\forall i=1,\dots,m$, where $\overline{c}$ is the maximum cost coefficient.   
\end{assumption}

The former requires OPF feasibility of the real transmission capacity data on all historical records, and the latter bounds the change in OPF costs to the cost of the most expensive unit. 


As a perturbed vector of line capacities may not be OPF feasible, we additionally introduce the relaxed OPF problem to give a numerical value to infeasibility of a particular vector $\overline{\varphi}$:
\begin{subequations}
\begin{align}
    \mathcal{C}_{i}^{R}(\overline{\varphi})=\minimize{p\in\mathcal{P}_{i},v\geqslant\mathbb{0}}\quad& c_{i}^{\top}p + \psi\mathbb{1}^{\top}v\label{OPF_rlxd:obj}\\
    \st\quad
    &\mathbb{1}^{\top}(p-d_{i})=0,\label{OPF_rlxd:bal}\\
    &\norm{F(p-d_{i})}_{1}\leqslant \overline{\varphi} + v\label{OPF_rlxd:flo},
\end{align}
\end{subequations} 
where the slack variable $v\in\mathbb{R}^{e}$ renders the OPF solution feasible for any assignment $\overline{\varphi}$. That is, infeasible vectors $\overline{\varphi}$ for problem \eqref{OPF_base} translate into costs $\mathcal{C}_{i}^{R}(\overline{\varphi})$ using penalty scalar $\psi\gg \overline{c}$.

\subsection{Differentially Private TCO Algorithm}

\begin{algorithm}[t]
\setlength{\abovedisplayskip}{0pt}
\setlength{\belowdisplayskip}{0pt}
\setlength{\abovedisplayshortskip}{0pt}
\setlength{\belowdisplayshortskip}{0pt}
\SetKwInOut{Input}{Input}
\SetKwInOut{Output}{Output}
\caption{Differentially private TCO for DC-OPF}\label{alg:OPF}
\Input{$\text{OPF dataset}\; \langle c_{i},d_{i},\mathcal{P}_{i},F,\overline{f}\rangle_{i=1}^{m}$;  DP parameters $\varepsilon_{1},\varepsilon_{2},\alpha$; iteration limit $T$}
\BlankLine 
\Output{%
    \xvbox{4mm}{Synthetic OPF data $\langle c_{i},d_{i},\mathcal{P}_{i},F,\overline{\varphi}\rangle_{i=1}^{m}$}
   }
\BlankLine 
{\it \underline{Step 1}:} Initialize synthetic dataset $\overline{\varphi}^{0} = \overline{f} + \text{Lap}\left(\alpha/\varepsilon_{1}\right)^{e}$
\BlankLine 
\For{$t\in1,\dots,T$ }{
\BlankLine 
{\it \underline{Step 2}:} Exponential mech. to find the worst-case model:

\For{$i\in1,\dots,m$ }{
    $
    \Delta\mathcal{C}_{i}=\norm{\mathcal{C}_{i}(\overline{f}) - \mathcal{C}_{i}^{R}(\overline{\varphi}^{t-1})}_{1} + \text{Lap}\left(\overline{c}\alpha/\varepsilon_{2}\right)$
}
\textbf{return:} index $k^{t}\leftarrow\text{argmax}_{i}\;\Delta\mathcal{C}_{i}$  of the worst-case model
\BlankLine
{\it \underline{Step 3}:} Laplace mechanism to compute the worst-case cost:
$$\overline{\mathcal{C}}_{t}=\mathcal{C}_{k^{t}}(\overline{f}) + \text{Lap}\left(\overline{c}\alpha/\varepsilon_{2}\right)$$
\BlankLine 
{\it \underline{Step 4}:} Post-processing optimization of synthetic data:
\begin{subequations}\label{OPF_PP}
\begin{align}
    \overline{\varphi}^{t}\in\argmin{\overline{\varphi}}\quad& \textstyle\sum_{\tau=1}^{t}\norm{\overline{\mathcal{C}}_{\tau} - \mathcal{C}_{k^{\tau}}(\overline{\varphi})} + \norm{\overline{\varphi}-\overline{\varphi}^{t-1}}\label{OPF_PP:obj}\\
    \st\quad&\text{DC-OPF}\;\eqref{OPF_base}\;\text{parameterized by}\;\overline{\varphi},\forall \tau \label{OPF_PP:con}
\end{align}
\end{subequations}
}
\textbf{return:} synthetic line capacity $\overline{\varphi}\leftarrow\overline{\varphi}^{T}$
\end{algorithm}

We now introduce the privacy-preserving transmission capacity obfuscation (TCO) Algorithm \ref{alg:OPF} for DC-OPF datasets. Here, 
Step 1 initializes synthetic dataset $\overline{\varphi}^{0}$ by perturbing real data using the Laplace mechanism, and the remaining steps post-process the synthetic dataset. Step 2 runs the report-noisy-max algorithm, a discrete version of the Exponential mechanism \cite{dwork2014algorithmic}, to privately identify the worst-case OPF model. Here, the score function, $\Delta\mathcal{C}$, is the $L_{1}$ norm which measures the distance between OPF costs on real and synthetic datasets.    
Then, Step 3 uses the Laplace mechanism to estimate the cost of the worst-case OPF model on the real data. 
Step 4 post-processes the synthetic dataset using a bilevel optimization problem \eqref{OPF_PP}, where $\mathcal{C}_{k^{\tau}}(\overline{\varphi})$ is the OPF costs obtained from the embedded DC-OPF problem \eqref{OPF_base} for some fixed vector $\overline{\varphi}$. By embedding the OPF problem as a constraint, we require feasibility and cost-consistency of $\overline{\varphi}$ with respect to the worst-case OPF models identified at previous steps. In addition, with the last term in \eqref{OPF_PP:obj}, we regularize the solution $\overline{\varphi}$ to make sure that the changes in synthetic capacities are only guided by feasibility and cost-consistency requirements. 

The major difference between the WPO and TCO algorithms is that the latter terminates after repeating Steps 2 to 4 $T$ times. The OPF feasibility for one model does not guarantee feasibility across the whole population of models. By increasing $T$, the TCO algorithm finds more worst-case OPF models with the largest cost $\mathcal{C}_{i}^{R}$ of violations, thereby improving the accuracy (feasibility) of the synthetic dataset across the population.

To arrive at a tractable mixed-integer reformulation of problem \eqref{OPF_PP}, we substitute constraint \eqref{OPF_PP:con} with the Karush--Kuhn--Tucker conditions of problem \eqref{OPF_base}; we refer to \cite[\S6]{pozo2017basic} for details. Notably, problem \eqref{OPF_PP} only relies on obfuscated data. Hence,  by Theorem \ref{th:PP}, it does not induce any privacy loss. We now state the $\varepsilon-$DP guarantee of the entire algorithm.

\begin{theorem}[DP of the TCO Algorithm]\label{th:DP_OPF} Setting $\varepsilon_{1}=\varepsilon/2$ and $\varepsilon_{2}=\varepsilon/(4T)$ renders Algorithm \ref{alg:OPF} $\varepsilon-$DP for $\alpha-$adjacent DC-OPF datasets.
\end{theorem} 
\begin{proof}See Appendix \ref{app:th_OPF}.
\end{proof}

\begin{remark}[Relation to prior work] When $m=T=1$, Step 2 in Algorithm \ref{alg:OPF} becomes redundant, and the algorithm replicates the Laplace-based PLO mechanism in \cite{fioretto2019differential}, when applied to the capacity obfuscation in the DC-OPF setting. The difference between the two algorithms reveals when the synthetic dataset must be accurate, i.e., feasible and cost consistent, on a population of OPF models, i.e., $m\gg1$. Indeed, the worst-case OPF model and cost can also be estimated using Laplace perturbations, but the induced privacy loss will reach $mT\varepsilon_{2}$. The combination of the Exponential and Laplace mechanisms in Steps 2 and 3 in Algorithm \ref{alg:OPF}, however, reduces the privacy loss to $2T\varepsilon_{2}$.
\end{remark}

\section{Numerical Experiments}\label{sec:experiments}

In our experiments, we fix the privacy loss $\varepsilon=1$ and vary adjacency parameter $\alpha$, thereby increasing the range of adjacent datasets, which are required to be statistically indistinguishable. All data and codes to replicate the results are available online:  
\begin{center}
    \url{https://github.com/wdvorkin/SyntheticData}
\end{center}

\subsection{Synthetic Wind Power Records Generation}

We first demonstrate the WPO Algorithm \ref{alg:WPO} for a privacy-preserving release of wind power records. We use the theoretical wind power curve of the General Electric GE-2.75.103 turbine from \cite{staffell2016using}, considering a medium range of wind speeds between $2.5$ and $12.5$ $\sfrac{m}{s}$, where the power output is most sensitive to speed. We then perturb each power output with a Gaussian noise $\mathcal{N}(0,0.1)$ to introduce some variation among the records; the dataset is thus not completely real, but resembles real-life datasets which we hope to eventually release with our algorithm. In the dataset, we have $m=1,000$ normalized power measurements $y\in[0,1]^{m}$ and corresponding wind speeds $x$. 

We specify regression \eqref{base_reg} as follows. First, we transform the wind speed records using $p=5$ Gaussian radial basis functions:
\begin{align*}
    \varphi_{j}(x) = e^{-\left(\frac{1}{2}\norm{x-\mu_{j}}\right)^2}, \forall j=1,\dots,p,
\end{align*}
positioned at $\mu=\{2.5,5,7.5,10,12.5\}\;\sfrac{m}{s}$. Each feature in $X$ is then obtained as $X_{ij}=\varphi_{j}(x_{i}),$ $\forall i=1,\dots,m,$ $ \forall j=1,\dots,p.$  Finally, we set the regularization parameter as $\lambda=10^{-3}$.

We use the standard Laplace mechanism as a reference method, which perturbs power records as $\tilde{y}=y+\text{Lap}(\alpha/\varepsilon)^{m}$, and projects them onto feasible range $[0,1]^{m}$. The resulting synthetic records satisfy $\varepsilon-$DP for $\alpha-$adjacent datasets, as per Theorems \ref{th:PP} and \ref{th:LM}. To guarantee $\varepsilon-$DP for the WPO algorithm, we set $\varepsilon_{1}$ and $\varepsilon_{2}$ according to Theorem \ref{th:DP_WPO}. We also set  regularization parameters $\gamma_{y},\gamma_{\beta}=10^{-5}$ for post-processing in \eqref{PP_reg}.

Figure \ref{fig:WPO_sum_plot} demonstrates some examples of synthetic wind power dataset releases. Here, we measure adjacency $\alpha$ in $\%$ of the nominal capacity of the wind turbine. Observe, that with increasing $\alpha$, the regression-agnostic Laplace mechanism yields a larger deviation of the synthetic records from the real data. While the WPO algorithm introduces even more noise, i.e., $\times3$ more noise at Step 1 and more noise at Step 2 due to sensitivities $\delta_{\ell}$ and $\delta_{\beta}$ growing in $\alpha$, the post-processing of the noisy records at Step 3 results in a better accuracy of the synthetic dataset. In Fig. \ref{fig:loss_summary}, we demonstrate the statistical significance of this observation by plotting the loss on synthetic datasets under the two methods. With increasing $\alpha$, the Laplace mechanism demonstrates a notable deviation from the loss on real data. The WPO algorithm, on the other hand, converges to the real loss on average, and does not significantly deviate from the average. 

\begin{figure*}
    \centering
    \includegraphics[width=0.95\textwidth]{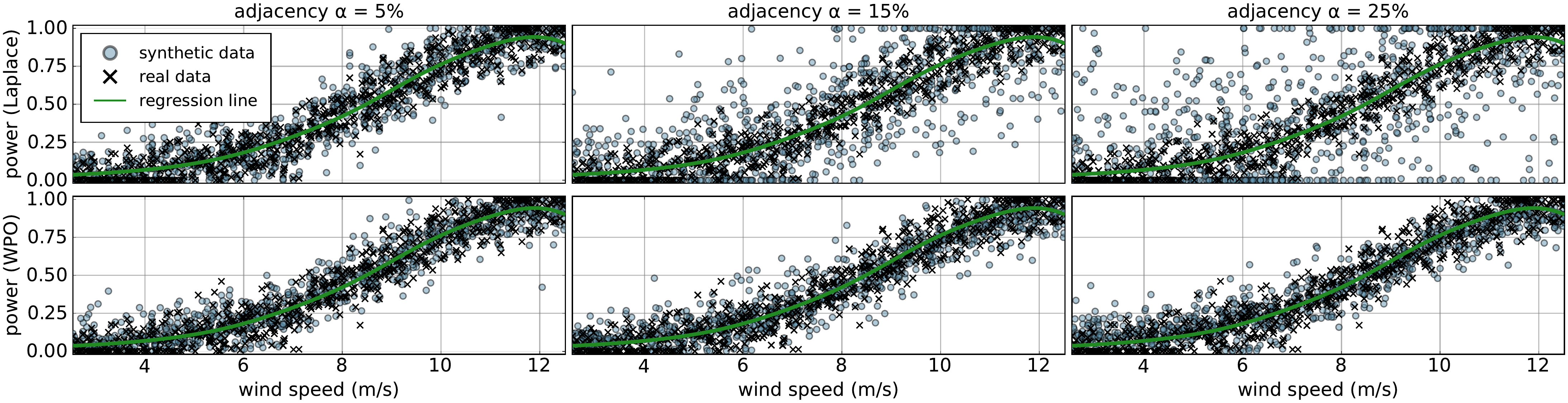}
    \caption{Wind power dataset obfuscation for the General Electric 2.75 MW turbine using the Laplace mechanism (top row) and the WPO algorithm (bottom row). }
    \label{fig:WPO_sum_plot}
\end{figure*}

\begin{figure}
    \centering
    \includegraphics[width=0.46\textwidth]{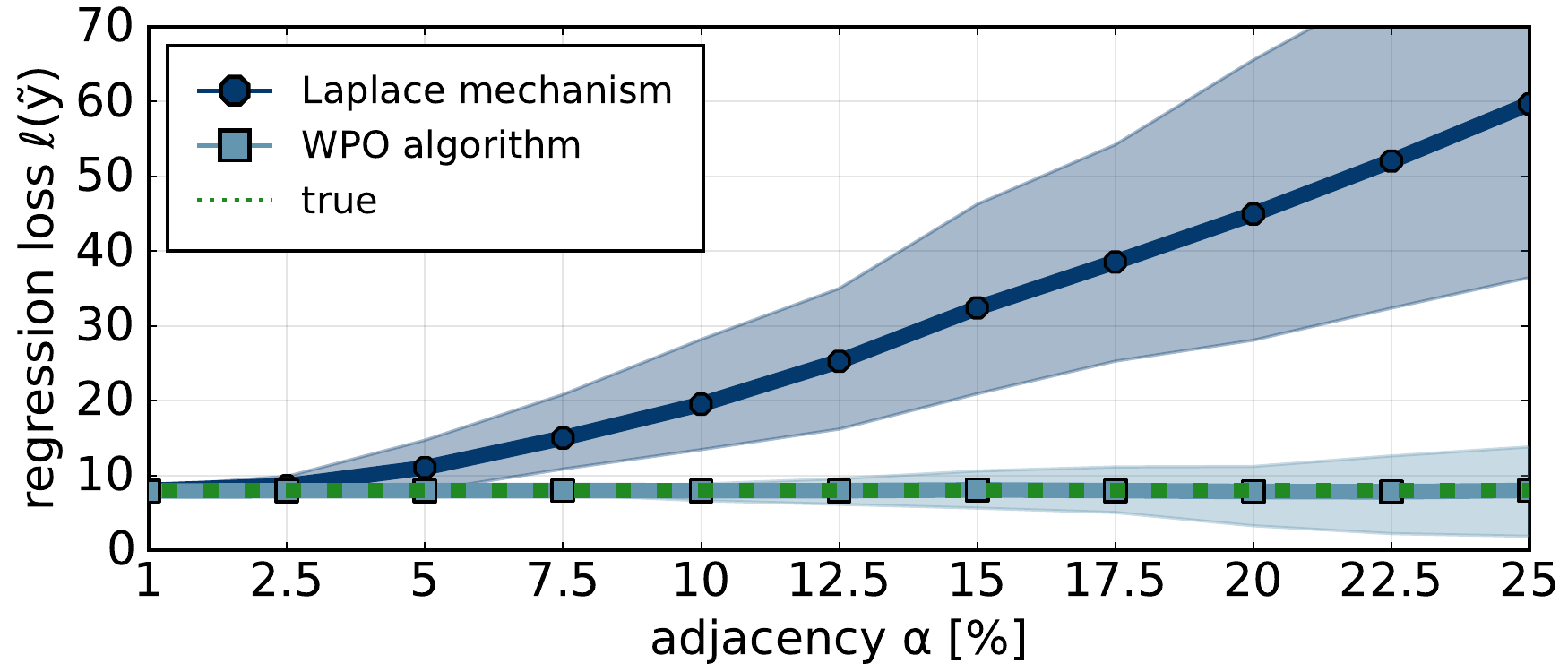}
    \caption{The mean and 90\% confidence band of the regression loss on synthetic datasets for 300 runs of the Laplace mechanism and WPO algorithm.}
    \label{fig:loss_summary}
\end{figure}


\subsection{Synthetic Transmission Data Generation}
We apply the TCO algorithm to a network data release from the IEEE 73-Bus Reliability Test System. To make the case more challenging, we reduce the transmission capacity to $60\%$ of the nominal level to increase network congestion.  We generate $m=10^3$ feasible DC-OPF datasets by sampling demand and generation limits from uniform distributions with bounds $\pm12.5$\% of their nominal values. The cost data is sampled from a uniform distribution $\mathcal{U}(80,100)$ \$/MWh, and we set penalty $\psi=3\cdot10^{3}$ in \eqref{OPF_rlxd:obj} for flow limit violations.  The privacy loss $\varepsilon$ is split according to Theorem \ref{th:DP_OPF}. Finally, we vary adjacency parameter $\alpha$ from 5 to 30 MW and iteration limit $T$ from 1 to 10. 

\begin{figure}
    \centering
    \includegraphics[width=0.46\textwidth]{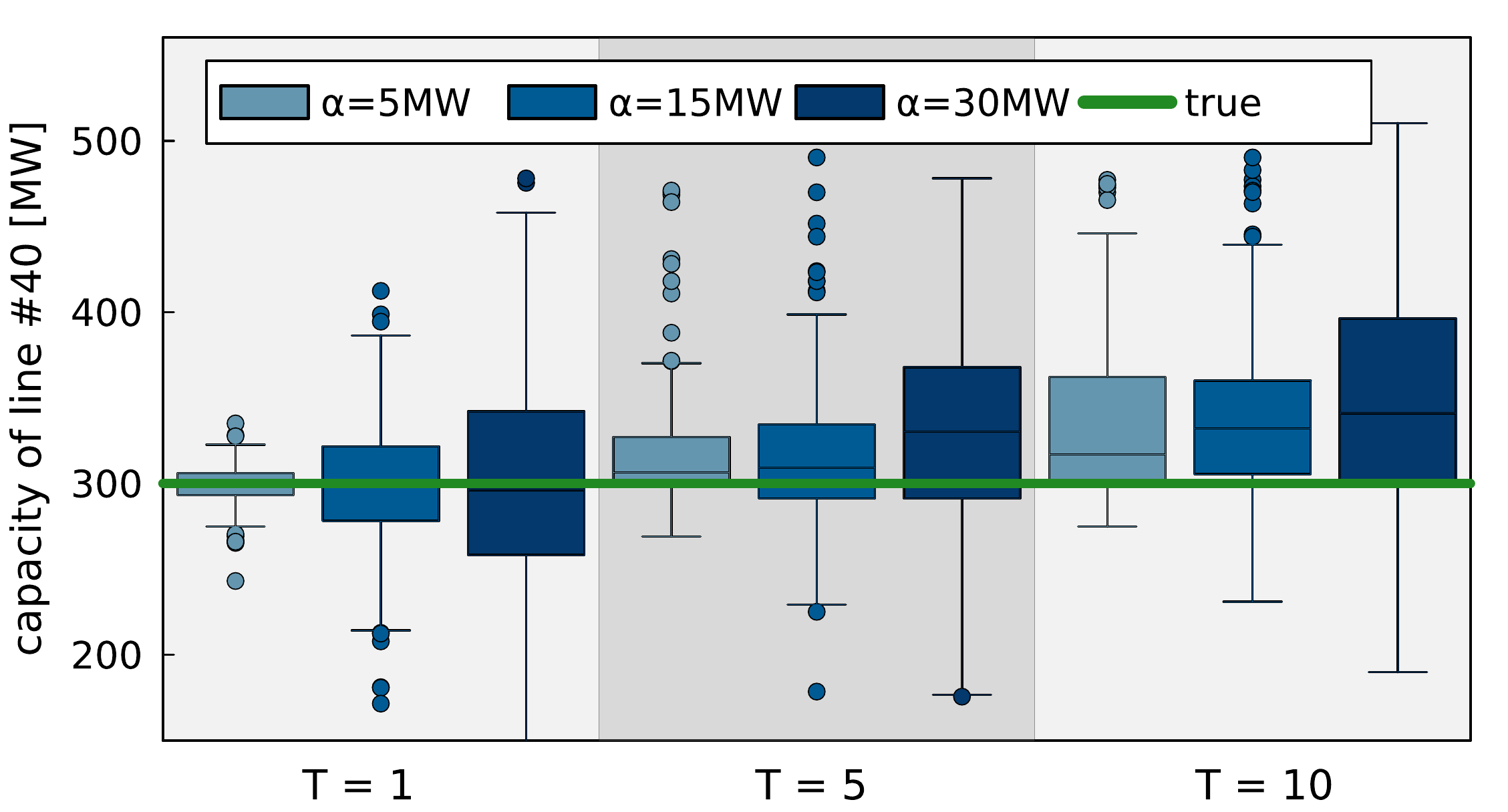}
    \caption{Distributions of obfuscation outcomes for line \#40 across 300 runs of the TCO algorithm for varying adjacency parameter $\alpha$ and iteration limit $T$.}
    \label{fig:TCO_line_40}
    \vspace{0.4cm}
\end{figure}

\begin{figure}
    \centering
    \includegraphics[width=0.48\textwidth]{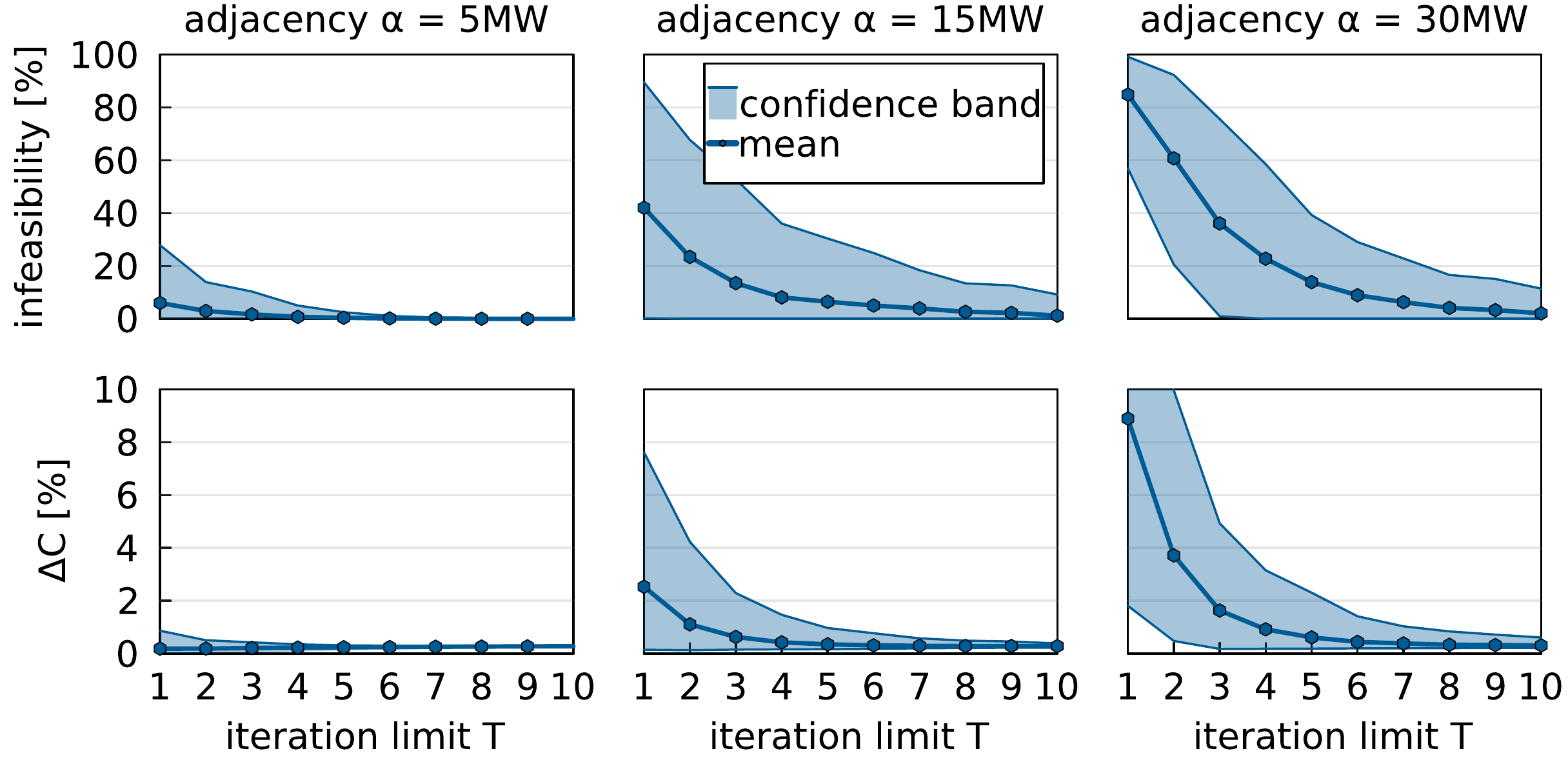}
    \caption{Infeasibility and sub-optimality of synthetic DC-OPF datasets. Top row: percentage of infeasible OPF solutions across a population of $m=1,000$ OPF models. Bottom row: the mean sub-optimality $\Delta\mathcal{C}$ of OPF costs on synthetic datasets. The mean values are provided with 90\% confidence bands.}
    \label{fig:OPF_feasibility_sub_optimality}
    \vspace{0.4cm}
\end{figure}

By increasing $\alpha$, we increase the noise magnitude at Step 1 of the TCO algorithm, resulting in a broader distribution of synthetic dataset outcomes, as depicted by box plots in Fig. \ref{fig:TCO_line_40} for one selected transmission line. However, as noise increases, the probability of obtaining an infeasible synthetic dataset also increases. We thus increase the iteration limit $T$ to improve the accuracy of the synthetic dataset. By setting $T$, we require feasibility and cost-consistency with respect to the set of $T$ worst-case OPF models and outcomes, provided at Steps 2 and 3, respectively. Such deeper post-processing results in distributional shifts, as further shown in Fig. \ref{fig:TCO_line_40} for increasing $T$. The virtue of these shifts is revealed in Fig. \ref{fig:OPF_feasibility_sub_optimality}, where the top row demonstrates how the probability of infeasible OPF outcomes on synthetic datasets reduces as the iteration limit increases. For smaller adjacency, it takes fewer iterations to restore the feasibility of the synthetic dataset. For example, for $\alpha=5$MW, it is enough to leverage $6$ worst-case OPF models in the post-processing optimization at Step 4 to restore feasibility across the entire population of $1,000$ OPF models. For larger adjacency parameters, it takes as much as 10 iterations on average. The bottom row in Fig. \ref{fig:OPF_feasibility_sub_optimality} depicts the mean sub-optimality of OPF models on the synthetic dataset, computed as: 
\begin{align}
    \Delta\mathcal{C}=\frac{1}{m}\sum_{i=1}^{m}\frac{\norm{\mathcal{C}_{i}({\overline{f}}) - \mathcal{C}_{i}^{R}({\overline{\varphi}^{T}})}}{\mathcal{C}_{i}({\overline{f}})}\times100\%.
\end{align}
The sub-optimality of synthetic datasets increases in adjacency parameter $\alpha$, as more noise corrupts the real data. However, as we increase $T$, the OPF cost on synthetic data gets closer to that on the real data. Eventually, the sub-optimality is kept very close to zero without violating the privacy of the real dataset.

\section{Conclusions}\label{sec:conclusions}

We developed two algorithms for privacy-preserving releases of synthetic wind power records and transmission capacity data. The former obfuscates power records by adding Laplacian noise to data and then post-processes the noisy data to privately restore accuracy using a reference machine learning model, thereby improving on the application of the Laplace mechanism alone. The latter enjoys both Laplace and Exponential mechanisms to release cost-consistent transmission data while ensuring the feasibility on a population of heterogeneous OPF models, without the need of drastically scaling the noise. Our results showed that identifying 10 worst-case OPF models suffices to restore data accuracy across the population of 1,000 models, on average.

\appendix
\subsection{Proof of Lemma \ref{lm:reg_sens}}\label{app:lemma1}
The worst-case sensitivity of regression weights is bounded as: 
\begin{subequations}
\begin{align}
    \delta_{\beta} &= \maximize{y \sim_{\alpha} y'}\;\norm{\beta(y) - \beta(y')}_{1}\\
    &=\maximize{y \sim_{\alpha} y'}\;\norm{(X^{\top}X + \lambda I)^{-1}X^{\top}(y -y')}_{1}\label{sens_b:one}\\
    &\leqslant\norm{(X^{\top}X + \lambda I)^{-1}X^{\top}}_{1}\cdot\maximize{y \sim_{\alpha} y'}\;\norm{y -y'}_{1}\label{sens_b:two}\\
    &\leqslant\norm{(X^{\top}X + \lambda I)^{-1}X^{\top}}_{1}\cdot\alpha\label{sens_b:three}
\end{align}
\end{subequations}
where equality \eqref{sens_b:one} is from the closed-form solution to the ridged regression, inequality \eqref{sens_b:two} is due to the  H\"{o}lders inequality, and inequality \eqref{sens_b:three} is due to Definition \ref{def:adjacency} of adjacent datasets. 

The sensitivity of regression loss $\ell$ is bounded as:
\begin{subequations}
\begin{align}
    \delta_{\ell}&=\maximize{y \sim_{\alpha} y'}\;\norm{\ell(y) - \ell(y')}_{1} \\
    &=\maximize{y \sim_{\alpha} y'}\;\norm{\norm{X\beta-y} -\norm{X\beta-y'}}_{1} \\
    &\leqslant\maximize{y \sim_{\alpha} y'}\;\norm{X(\beta(y) - \beta(y'))-(y - y')}\label{sens_l:one}\\
    &=\maximize{y \sim_{\alpha} y'}\;\norm{(X(X^{\top}X + \lambda I)^{-1}X^{\top} - I)(y-y')}\label{sens_l:two}\\
    &=\maximize{i=1,\dots,m}\;\norm{(X(X^{\top}X + \lambda I)^{-1}X^{\top} - I)(e_{i}\circ\alpha)}\label{sens_l:three}
\end{align}
\end{subequations}
where inequality \eqref{sens_l:one} is due to the reverse triangle inequality, equality \eqref{sens_l:two} is from the closed-form solution to the ridged regression. Equality \eqref{sens_l:three} originates from Definition \ref{def:adjacency} of adjacent datasets, i.e., different in one element by at most $\alpha$. It is thus enough to find index $i$ of that element which maximizes the norm. 

\subsection{Proof of Theorem \ref{th:DP_WPO}}\label{app:th_WPO}

The algorithm queries real data in the following computations: 
\begin{enumerate}
    \item Initialization at Step 1 using the Laplace mechanism with parameters $\alpha/\varepsilon_{1}$. Since the worst-case sensitivity of identity queries is $\alpha$ \cite{chatzikokolakis2013broadening}, this computation is $\varepsilon_{1}$-DP by Theorem \ref{th:LM}. 
    \item Estimation of the regression loss on the real data at Step 2 using the Laplace mechanism with parameters $\delta_{\ell}/\varepsilon_{2}$. By Lemma \ref{lm:reg_sens} and Theorem \ref{th:LM}, this estimation is $\varepsilon_{2}-$DP. 
    \item Estimation of regression weights on the real data at Step 2 using the Laplace mechanism with parameters $\delta_{\beta}/\varepsilon_{2}$. By Lemma \ref{lm:reg_sens} and Theorem \ref{th:LM}, this estimation is $\varepsilon_{2}-$DP. 
\end{enumerate}
Note, that the post-processing optimization at Step 3 only uses obfuscated data. Hence, it does not induce any privacy loss per Theorem \ref{th:PP}. Per Theorem \ref{th:SC}, the total privacy loss becomes $\varepsilon_{1} + 2\varepsilon_{2}$, yielding $\varepsilon$ when setting parameters $\varepsilon_{1}=\varepsilon/2$ and $\varepsilon_{2}=\varepsilon/4$.

\subsection{Proof of Theorem \ref{th:DP_OPF}}\label{app:th_OPF}
We follow similar arguments. Algorithm \ref{alg:OPF} queries private transmission capacity vector $\overline{f}$ for the following computations:
\begin{enumerate}
    \item Initial dataset $\overline{\varphi}^{0}$: the algorithm uses a private identity query with privacy budget $\alpha/\varepsilon_{1}$. Since the sensitivity of identity queries on $\alpha-$adjacent datasets is $\alpha$ \cite{chatzikokolakis2013broadening}, by Theorem \ref{th:LM} this computation is $\varepsilon_{1}-$DP.
    \item Worst-case OPF index: found by the discrete variant of the Exponential mechanism with privacy budget $\overline{c}\alpha/\varepsilon_{2}$. Since the sensitivity of the score function $\Delta\mathcal{C}_{i}$ is the same as that of $\mathcal{C}_{i}$, by Theorems \ref{th:PP} and \ref{th:EM} and Assumption \ref{ass:sensitivity}, this is $\varepsilon_{2}-$DP. 
    \item Worst-case OPF cost: Step 3 uses a private identity query of the worst-case OPF cost using privacy budget $\overline{c}\alpha/\varepsilon_{2}$. Per Assumption \ref{ass:sensitivity} and Theorem \ref{th:LM}, this computation is $\varepsilon_{2}-$DP. 
\end{enumerate}
Let $\overline{\varepsilon}$ be the total privacy loss accumulated by the algorithm. Step 1 accumulates privacy loss of $\varepsilon_{1}$. Since Steps 2 and 3 repeat $T$ times, per Theorem \ref{th:SC}, they accumulate the privacy loss of $2T\varepsilon_{2}$. The total loss is then $\overline{\varepsilon} = \varepsilon_{1} + 2T\varepsilon_{2},$ which amounts to $\varepsilon$ when setting DP parameters $\varepsilon_{1}=\varepsilon/2$ and $\varepsilon_{2}=\varepsilon/(4T)$.

\bibliographystyle{IEEEtran}
\bibliography{references}

\endgroup
\end{document}